\documentclass[a4paper]{article}[12pt]
\usepackage{gamesem}
\usepackage{todonotes}
\usepackage{amsmath, amsthm, amssymb}
\usepackage[defblank]{paralist}
\usepackage{shadowbox}
\usepackage{manfnt}

\newtheorem{theorem}{Theorem}[section]

\newtheorem{lemma}{Lemma}[section]
\newtheorem{proposition}{Proposition}[section]

\theoremstyle{remark}
\newtheorem{remark}{Remark}[section]
\newtheorem{example}{Example}[section]

\newtheorem{convention}{Convention}[section]

\theoremstyle{definition}
\newtheorem{definition}{Definition}[section]

\newcommand\orddec{{\sf orddec}}

\newcommand\Nodes{N}

\newcommand\arity{\mathop{\mathrm{arity}}}

\author{William Blum}
\title{Type homogeneity is not a restriction for safe recursion schemes}

\begin{document}
\maketitle
\begin{abstract}
Knapik \emph{et al.\ }introduced the \emph{safety restriction} which constrains both the types and syntax of the production rules defining a higher-order recursion scheme \cite{KNU02}. This restriction gives rise to an equi-expressivity result between order-$n$ pushdown automata and order-$n$ safe recursion schemes, when such devices are used as tree generators.

We show that the typing constraint of safety, called \emph{homogeneity}, is unnecessary in the sense that imposing the syntactic restriction alone is sufficient to prove the equi-expressivity result for trees.
\end{abstract}

\section{Introduction}
The concept of \emph{safety} was introduced by Knapik \emph{et al.} in the
context of higher-order recursion schemes and studied further in \cite{KNU02,
demirandathesis}. In the original paper they show that, when used as tree
generators, safe recursion schemes are equivalent to safe pushdown automata
\cite{KNU02}.
In its original definition, safety consists of two fairly technical constraints:
\begin{enumerate}
\item \textbf{Type-Homogeneity} which imposes a type constraint on the rules of
the recursion scheme;
\item \textbf{A syntactic restriction} on the relative order of sub-terms
occurring in the right hand side of the recursion scheme rules.
\end{enumerate}

Knapik \emph{et al.}'s result begged a question: what is the automaton
equivalent of `unsafe` recursion schemes?
This question was answered a few years later in a paper introducing a new kind
of automata called Higher-Order Collapsible Pushdown Automata (CPDA), which is
shown to be equivalent to (possibly unsafe) higher-order recursion schemes
\cite{hmos-lics08}. More specifically they describe an algorithm that, given an
order-$n$ recursive schemes $G$, constructs an equivalent order-$n$ CPDA (in the
sense that they both generate the same tree), and conversely. We will refer to
them as the \emph{HMOS transformation} procedures.

The safety restriction was further studied in the context of the lambda calculus
in the author's D.Phil.\ thesis \cite{BlumPhd}. Somewhat surprisingly, in the
lambda calculus setting, the type-homogeneity constraint is not necessary to
define a useful calculus. One can define a notion of safe simply-typed lambda
calculus by imposing solely the syntactic constraint above, that we name
`\emph{incremental-binding}', onto the standard simply-typed lambda calculus.
The author's D.Phil.\ thesis gives expressivity results and presents a fully-abstract game
model for this calculus \cite{BlumPhd}.

This paper reconciles the above observation made in the lambda calculus with higher-order recursion schemes: we show that if a recursion scheme meets the syntactic criteria of the safety restriction, then regardless of whether the type homogenity is met, the automaton constructed using the HMOS procedure from \cite{hmos-lics08} can be converted into an equivalent order-$n$ (non-collapsible) push-down automaton (PDA). Conversely, given an order-$n$ PDA, the recursion scheme generated by the HMOS transformation induces an equivalent incrementally-bound and homogeneously-typed recursion scheme.
Consequently, for generating trees, pushdown automata are just as expressive as incrementally-bound recursion schemes.
(In what follows, by abuse of language, we will use ``\emph{incremental-binding}" to refer to the syntactic constraint 2.\ above, even though the concept of binders is foreign to recursion schemes.)

Further, composing the above two transformations yields a methods to
convert any incrementally-bound recursion scheme into an equivalent \emph{safe} one in the original sense ({\it i.e.,} incrementally-bound and type-homogeneous). Type-homogeneity is therefore not a proper restriction for safe recursion schemes. This generalizes Knapik \emph{et al.}'s result on equi-expressivity of pushdown automata and \emph{safe} recursion schemes \cite{KNU02}.
\vspace{1em}

\begin{remark}
The result presented in this paper was privately circulated for the first time in 2009 and shared on my personal website but was never published in a journal or conference \cite{blum-safehomogeneity-note}.
The result was first conjectured in the author's thesis in \cite{BlumPhd}. A partial proof was then privately circulated by the author in a 2008 note. Based on this note, Broadbent \cite{Broadbent2009} adjusted the definition of \emph{stack safety} to make the inductive proof work, which filled the remaining gap in the proof. The author then circulated an updated proof based on yet another definition of \emph{stack safety} which is also the one used in the present paper.
\end{remark}

\section{Background}
We first recall some basic notations and recall the definition of higher-order automata and higher-order recursion schemes from the literature \cite{KNU02,hague-collaps-full}.

We consider simple types over a single atom $o$. We call $\Sigma$ a \defname{ranked alphabet} if each symbol $f \in \Sigma$ is assigned a type of the form $o \longrightarrow  \cdots \longrightarrow o \longrightarrow  o$. We write $\arity(f)$ for any typed-term $f$ to denote the arity of its type.

Given a set of typed symbols $S$, the set of \defname{applicative terms} generated from $S$, is defined as the closure of $S$ under the application rule: if $m$ and $n$ are two applicative terms of respective type $A\rightarrow B$ and $A$ then $(m n)$ is also an applicative term of type $B$.

A \defname{$\Sigma$-labelled tree} is a possibly infinite tree where each nodes of the tree is labelled with a symbol in $\Sigma$ with arity matching the number of children nodes in the tree.

\subsection{Higher-order stacks and recursion schemes}

\paragraph{Higher-order stacks}
The notion of higher-order stacks is defined inductively. We first fix the base alphabet as a finite set $\Gamma$. An \emph{order-$1$} stack over $\Gamma$ is a sequence of elements of $\gamma \in \Gamma^*$. For $n \in \mathbb{N}$ an \emph{order-$(n+1)$}-stack is defined as a stack of order-$n$ stacks.
The order of a stack $s$ is written $\ord(s)$. The $i^{th}$ elements in a stack $s$, for $i\geq 0$, is denoted $s_i$ so that a stack $s$ consists of ordered elements $s_1$, $\dots$ $s_k$ where $k\geq0$ is called the length of the stack. The stack can then be represented by the notation $[s_1, \cdots s_k]$. The first element $s_1$ and last element $s_k$ are respectively called the `bottom' and `top' elements of the stack.
The empty stack of order $1$ is denoted $\bot_1$ and represent the stack consisting of no element. We then define the empty $(n+1)$-stack $\bot_{n+1}$ as the singleton sequence $[ \bot_n ]$.

The following operations are defined for an order-$1$ stack $s$:
\begin{eqnarray*}
	push_1^{a} [ \; s_1 \cdots s_m \; ] & := & [ \; s_1 \cdots s_m\, a \; ] \\
	pop_1 [ \; s_1 \cdots s_m s_{m+1} \; ] & : = & [ \; s_1 \cdots s_m \; ] \\
\end{eqnarray*}
and for an order-$(n+1)$ stacks $t$:
\begin{eqnarray*}
	push_{n+1} [ \; t_1 \cdots t_m \; ] & := & [ \; t_1 \cdots t_m\; t_m \; ] \\
	pop_{n+1} [ \; t_1 \cdots t_m t_{m+1} \; ] & : = & [ \; t_1 \cdots t_m \; ].
\end{eqnarray*}

The $top_i$ of a stack for $i\geq 1$ is defined as the top $(i-1)$-stack of $s$ (so that $pop_i$ returns $s$ with its top $(i-1)$-stack removed).
Formally, for $s=[ \; s_1 \cdots s_{l+1} \; ]$, $l\geq0$ and $1 \leq i\leq \ord(s)$:
\begin{eqnarray*}
	top_{i} [ \; s_1 \cdots s_{l+1} \; ] & := &
	\begin{cases}
		s_{l+1} & \mbox{if $i = \ord(s)$} \\
		top_i s_{l+1} & \mbox{if $i < \ord(s)$.}
	\end{cases}
\end{eqnarray*}

Following \cite{hmos-lics08}, we define the natural notion of stack prefix: for a given $n$-stack $s$, for any lower-level stack $m$ occurring in $s$, the prefix $s_{\prefixof m}$ of $s$ at $m$ consists of the stack obtained from $s$ by removing all the elements `above' $m$ using a succession of `pop` operations.

Formally:
\begin{definition}[Stack prefix]
Let $s = [s_1 \cdots s_m]$ be a higher-order stack. Then $s_{\prefixof s_i}$ is defined as $[s_1 \cdots s_i]$ for $1 \leq i \leq m$; and for a stack $t$ occurring in $s_i$ we define $s_{\prefixof t}$ as $[s_1 \cdots {s_i}_{\prefixof t}]$.

Similarly we define the strict stack prefix ${s}{< s_i}$ and $s_{<t}$ respectively as
 $[s_1 \cdots s_{i-1}]$ and $[s_1 \cdots t_{< s_{i}}]$.
\end{definition}

For any operation $\phi$ operating on a higher-order stack and $k\geq1$ we denote $\phi^k$ the repeated application of $\phi$ exactly $k$ times: $\phi^k\; s = \phi ( \cdots (\phi\ s) \cdots)$ with $k$ application of $\phi$ for any stack $s$.

\paragraph{Higher-order stacks with links}

We recall the notion of higher-order stack with links where each order-$1$ element is equipped with a link to another stack \cite{hague-collaps-full}:

\begin{definition}[Higher order stack with links (or CPDA stack)]
For a stack alphabet $\Gamma$ and a distinguished bottom-of-stack symbol $\bot\in \Gamma$. An order-$0$ CPDA stack is just a stack symbol. An order-$(n+1)$ CPDA stack $s$ is a non-empty sequence (written $[s_1 \cdots s_l]$) of order $n$ CPDA stacks such that every non-$\bot$ symbol $a$ that occurs in $s$ has a link to a stack of some order $k$ (where $0 \leq k \leq n$) situated below it in $s$; we call the link a $(k + 1)$-link.
\end{definition}

An element of a higher order CPDA stack is written $a^{(j,k)}$ where $a\in \Gamma$ and the exponent $(j,k)$ encodes the pointer associated to the stack symbol. Think of it as a shorthand for the iterated stack operation $pop_j^k$ ({\it i.e.,} $k$ applications of $pop_j$).
The value $j$ is called the \defname{order of the link}, and $k$ is called the \defname{height of the link} for $1 \leq j \leq n$ and $k \geq 1$, such that if $j = 1$ then $k = 1$.

\begin{definition}
Let $s$ be a higher-order stack. We define $s^{\langle j \rangle}$ as the operation that replaces
every link occurring in $s$ of the form $(j,k)$ with $(j,k+1)$. Formally,
\begin{align*}
{a^{(j,k)}}^{\langle j \rangle} &= a^{(j,k+1)}   \\
{a^{(j,k)}}^{\langle j' \rangle} &= a^{(j,k)} &   \mbox{when $j\neq j'$,}\\
[s_1 \ldots s_p]^{\langle j \rangle} &= [s_1^{\langle j \rangle} \ldots s_p^{\langle j \rangle}] \ .
\end{align*}
\end{definition}
This operation clearly commutes with prefixing. We will therefore write $s^{\langle j\rangle}_{\prefixof m}$ as an abbreviation for
$(s^{\langle j\rangle})_{\prefixof m} = (s_{\prefixof m})^{\langle j\rangle}$,
for all stack $s$ and stack-symbol $m$ occurring in $s$.

The $top$ and $pop$ operations are defined identically to standard higher-order stacks.
The $push$ operation is defined similarly to higher-order stacks except that the push operation now assigns a pointer to every element pushed onto the stack. Also, in addition to the standard push and pop operations, the CPDA introduces a new $collapse$ operation that `collapses' a given stack to the target of the pointer associated with the top order-$1$ element in the stack. The idea is that if the top-$1$ element is a symbol $a$ with a link to some $(j - 1)$-stack, then the $collapse$ operation produces the same effect as successively performing $k$ times the $pop_j$ operation.

We reproduce here the definition of the $push$ and $collpase$ operations from \cite{hmos-lics08}:

\begin{definition}[Higher-order CPDA operations]
	\label{cpdastackopearations}
We define CPDA operations in terms of the standard stack operations of an order-$n$ PDA with an adequately defined alphabet encoding elements of the form $b^{(o,h)}$ where $b\in\Gamma$ and link indices $o,h \geq 0$. For $1 \leq i  \leq \ord(s)$ and $2  \leq j  \leq ord(s)$:
\begin{align*}
push_1^{b,i}\ s &= push_1^{b^{(i,1)}}\ s \\
collapse\ s &= pop_o^h s  \hbox{\ where } top_1\ s = a^{(o,h)} \\
push_j \underbrace{[ s_1 \ldots s_{l+1} ]}_{s} &=
\left\{
\begin{array}{ll}
\    [s_1\ \ldots\ s_{l+1}\ s_{l+1}^{\langle j \rangle}]  &\hbox{if $j = \ord{s}$;}\\
\    [s_1\ \ldots\ s_{l+1}\ push_j\ s_{l+1}]  &\hbox{if $j<\ord{s}$.}
\end{array}
\right.
\end{align*}
\end{definition}

\begin{convention}
	In the rest of the paper we will adopt the following abbreviations:
\begin{itemize}
	\item ``$push_1^{a,j}$'' for the operation $push_1\,a^{(j,1)}$;
	\item ``$push_1^a$'' for the operation $push_1\,a^{(j,k)}$ where the components $j$ and $k$ are undetermined.
\end{itemize}
\end{convention}

\begin{convention}[Top stack convention]
\label{conv:top_preserve_links}
In the original definition, the operation $top_i$, that returns the top $(i-1)$-stack of a higher-order stack,
 removes any dangling pointer resulting from the operation. Here, we suppose that $top_i$ is defined in such a way that all pointers are preserved. From an implementation viewpoint, since links are encoded as pairs of integers, this means that $top_i$ just returns an unmodified copy of the top $(i-1)$-stack.
\end{convention}

\begin{definition}
	A tree generating \defname{order-$n$ collapsible pushdown automaton}, $n$-CPDA for short, is a tuple $\langle \Sigma, \Gamma,Q,\delta,q_0 \rangle$ where $\Sigma$ is a ranked alphabet,
	$\Gamma$ is a stack alphabet, $Q$ is a finite set of states, $q_0$
	is the initial state, and $\delta$ is a transition function $ Q \times \Gamma Γ \longrightarrow Q \times Op + Out$ where
	\begin{itemize}
		\item  $Op$ denotes the set of operations on higher-order stacks with links from Definition \ref{cpdastackopearations}.;
		\item $Out = \{(f, q_1 , \cdots ,q_{\arity(f)})  : f \in \Sigma, q_ i \in Q \}$ denotes the set of possible output emitted at each step of the transition function.
	\end{itemize}
	An order-$n$ CPDA stack is called a \defname{configuration} of an order-$n$ CPDA.
\end{definition}

This definition extends the notion of tree-generating \defname{Pushdown Automaton} (PDA) which are similarly defined on regular higher-order stack with no links and without the $collapse$ stack operation.

\paragraph{Tree accepted by a CPDA/PDA}

One can think of a CPDA (resp.\ PDA) as a state machine equipped with a higher-order stacks. At each transition, the CPDA can either (i) applies some operation to the stack and update its state, (ii) output the root node $f\in \Sigma$ of the generated tree together with new states for each branch of the tree root. The \defname{infinite tree accepted} by the automaton can then be obtained by recursively spawning new automatons at each child of the node starting with the specified state $q_i$.
The formal definition of the generated tree can be found in \cite{KNU02,hague-collaps-full}.

\paragraph{Higher-order recursion schemes} \hfill

\begin{definition}
	A \defname{higher-order recursion scheme} is a tuple $\langle
	\Sigma, \mathcal{N}, \mathcal{R}, S \rangle$, where $\Sigma$ is a
		ranked alphabet of \emph{terminals};
		$\mathcal{N}$ is a finite set of typed \emph{non-terminals};
		$S$ is a distinguished ground-type symbol of
		$\mathcal{N}$, called the start symbol;
		$\mathcal{R}$ is a finite set of production rules.
		For each non-terminal $F : (A_1, \ldots, A_n, o) \in \mathcal{N}$ there is exactly one rule of the form:
		$ F z_1 \ldots z_m \rightarrow e$
		where each $z_i$ (called \emph{parameter}) is a
		variable of type $A_i$ and $e$ is an applicative term of type $o$
		generated from the typed symbols in $\Sigma \union \mathcal{N} \union \{z_1:A_1, \ldots, z_m:A_m \}$.

	We say that the recursion scheme is \emph{order-$n$} just in case the order of the highest-order non-terminal is $n$.

\end{definition}

\paragraph{Tree-generated by a recursion scheme}
An applicative term generated from the terminal symbols $\Sigma$ only (without non-terminals), is viewed as a $\Sigma$-labelled tree and called a \defname{value tree}. A recursion scheme defines a (potentially infinite) value tree obtained by unfolding its rewrite rules \emph{ad infinitum}, replacing formal by actual parameters each time, starting from the start symbol $S$. It is formally defined as the least upper bound of the
induced \emph{schematological tree grammar} in the continuous algebra of ranked trees with the appropriate ordering \cite{KNU02,demirandathesis}.

\parpic[r]{
	\raisebox{-30pt}
	{\begin{tikzpicture}[baseline=(root.base),level distance=3ex,inner ysep=0.5mm,sibling distance=13mm]
		\node (root) {$g$}
		child {node {$a$}}
		child {node {$g$}
			child {node {$a$}}
			child {node {$h$}
				child {node {$h$}
					child {node {$\ldots$}}
				}
			}
		};
		\end{tikzpicture}
	}
}
\begin{example}\label{eg:running}
	Let $G$ be the following order-2 recursion scheme:
	\[\begin{array}{rll}
	S & \rightarrow & H \, a\\
	H \, z & \rightarrow & F \, (g \,
	z)\\
	F \, \phi & \rightarrow & \phi \, (\phi \, (F \, h))\\
	\end{array}\]
	with non-terminals $S:o$, $F : ((o, o),o)$, $H:(o,o)$ and terminals $g, h, a$ of arity $2, 1, 0$ respectively.
	Then the tree generated by $G$ is defined by the infinite term
	$g \, a \, (g \, a \, (h \, (h \, (h \,
	\cdots))))$ pictured on the right.
	
\end{example}

\subsection{Computation tree}

Informally, the \emph{computation tree} of a higher-order recursion scheme from \cite{OngLics2006} is defined as a finite tree representation of the \emph{$\eta$-long normal form} of the (possibly infinite) lambda terms inherent to the production rules of the recursion scheme.

Fix a higher-order recursion scheme $R = \langle \Sigma, \mathcal{N}, \mathcal{R}, S \rangle$.
Observe that production rules naturally map to simply-typed lambda terms by (i) currying the rule: replacing variables defined on the left-hand side by a succession of lambda abstractions on the right-hand side, (ii) interpreting terminals and non-terminal occurring in the applicative term of the right-hand side as free variables. Supposed that $F : A_1 \rightarrow \cdots \rightarrow A_m \in \mathcal{N}$ for some $m\geq 0$. Then the rule $F z_1 \ldots z_m \rightarrow f (G (F z_1))$ corresponds to the simply-typed lambda-term:
$$\lambda z_1 \ldots z_m . f (G (F z_1))$$ of type $A_1 \rightarrow \cdots \rightarrow A_m \rightarrow o$ with free variables $f, F, G$ of appropriate type. For each non-terminal $F$ we denote this term $\Lambda(F)$.

We recall that a simply-typed lambda term is in \defname{$\eta$-long normal form} if it can be written $\lambda \overline{x} . s_0 s_1 \ldots s_m$ (where $\lambda \overline{x}$ is an abbreviation for $\lambda x_1 \ldots \lambda x_n$ for some $n\geq 0$) where $s_0 s_1 \ldots s_m$ is of ground type, each $s_j$ for $j\in 1..m$ is in $\eta$-long normal form, and either $s_0$ is a variable and $m\geq0$; or $s_0$ is an abstraction $\lambda\overline{y}.s$ for some $\eta$-long normal form $s$ and $m\geq1$.
It is convenient to write the $\eta$-long normal form of terms of ground type as $\lambda . N$ for some term $N$ where `$\lambda.$' is referred to as a `dummy lambda'.

Observe that a purely applicative term can trivially be converted to $\eta$-long normal form by recursively $\eta$-expanding every subexpression.

We define the computation tree of a simply-typed term as an abstract syntax tree of its eta-long normal form:
\begin{definition}[Computation tree of a simply-typed term (\cite{BlumPhd})]
	Let $M$ be a simply-typed term of type $T$ with variables names $\mathcal{V}$.
	The \defname{computation tree} $\tau(M)$ with labels taken from $ \{ @ \} \union \mathcal{V} \union \{ \lambda x_1 \ldots x_n \ | \ x_1 ,\ldots, x_n \in
	\mathcal{V}, n\in\nat \}$, is defined inductively on the $\eta$-long normal form of $M$ as follows.
	\begin{eqnarray*}
		\tau(\lambda \overline{x} . z s_1 \ldots s_m) &=& \underline{\lambda \overline{x}}\; \langle\; \underline z\; \langle\tau(s_1),\ldots,\tau(s_m)\rangle\rangle\\
		&& \hbox {where $m\geq 0$, $z \in \mathcal{V}$}, \\
 \tau(\lambda \overline{x} . (\lambda y.t) s_1 \ldots s_m) &=& \underline{\lambda \overline{x}}\; \langle\; \underline @ \; \langle \tau(\lambda y.t),\tau(s_1),\ldots,\tau(s_m) \rangle \rangle \enspace\\
&&  \hbox{where $m\geq 1$, $y\in \mathcal{V}$}.
	\end{eqnarray*}
	where $l\langle t_1, \ldots, t_n \rangle$ for $n \geq 0$ succinctly denotes a labelled tree with root labelled $l$ and $n$ ordered children trees $t_1$, \ldots, $t_n$. The underlined expressions correspond to the nodes of the tree for $m>0$, and leaves for $m=0$.
\end{definition}

We borrow terminology from the lambda calculus when referring to the computation tree. In particular we will sometime refer to the `binder` of a variable node as the uniquely defined lambda node that binds the variable in the corresponding lambda term.

Given a simply-typed lambda term $M$ with free variables in $\mathcal{N} \union \Sigma$, we define the \defname{unfolding of $M$} as the simply-typed lambda term $M$ obtained by replacing every variable $N \in \mathcal{N}$ by the term $\Lambda (N)$ using capture-permitting substitution. The unfolding of the computation tree $\tau(M)$ is defined as the computation tree of the unfolding of $M$.

\begin{definition}
The (possibly infinite) \defname{computation tree of a recursion scheme} $R$ is defined as the recursive unfolding
of the computation tree of the simply-typed term $\Lambda(S)$.

(An alternative definition can be found in \cite{OngLics2006}.)
\end{definition}
\begin{remark} The repeated unfolding operation does not incur capture of variables since the lambda terms representing rewrite rules of a recursion scheme only have free variables in $\mathcal{N} \union \Sigma$, and leave all variables in $\mathcal{N} \union \Sigma$ unbound.
\end{remark}

\paragraph{Incremental binding}
We recall that the order of a simple type is defined as $\ord(o) = 0$ for any base type $o$, and $\ord(A \typear B) = \max(\ord(A)+1, \ord(B))$. We define the order of a variable node as the order of its associated simple type, and the order of a lambda node $\lambda x_1 \cdots x_n$ for $n>0$ as $1 + \max_{0\leq i\leq n} \ord{x_i}$.

We now recall the following notion from the author's thesis \cite{BlumPhd,blumong:safelambdacalculus} which relates to the syntactic restriction of safety:
\begin{definition}[\cite{BlumPhd,blumong:safelambdacalculus}]
The computation tree of a closed lambda term is \defname{incrementally-bound} if every variable node $x$ is bound by the first lambda node in the path from it to the root that has order strictly greater than $x$.

By extension, we say that the computation tree of an open term $M$ with free variables $x_1, \cdots x_n$, $n\geq 0$ is incrementally-bound if so is the computation tree of closed term $\lambda x_1 . \cdots \lambda x_n . M$.
\end{definition}

We extend this definition to recursion schemes through finite approximation of the computation tree: a recursion scheme is \defname{incrementally-bound} just if every finite unfolding of the tree is incrementally-bound. (Or equivalently, if every finite tree obtained by pruning some branches of its computation tree is incrementally-bound.)

\subsection{The safety restriction}
\label{sec:safety}

The safety restriction has appeared under different forms in the literature \cite{KNU02,Dam82,demirandathesis,BlumPhd,blumong:safelambdacalculus}.
We include here a reformulation of the definition by Knapik {\it et al.} in the setting of recursion schemes \cite{KNU02}.

\begin{definition}[Type homogeneity]
\label{subsec:typehomogeneity}
We say that a simple type $A_1 \longrightarrow \cdots \longrightarrow A_n \longrightarrow o$ with $n \geq 0$ is \defname{homogeneous}
just if $\ord{A_1} \geq \ord{A_2}\geq \cdots \geq \ord{A_n}$, and each $A_1$, \ldots, $A_n$ is homogeneous \cite{KNU02}.
\end{definition}

\begin{definition}[Safe recursion scheme]
\label{def:safers}
Let $G$ be a higher-order recursion scheme where the \emph{non-terminals all have homogeneous types}.
We say that $G$ is \textbf{unsafe} just if it has a production rule $F z_1 \ldots z_m \rightarrow e$ where $e$ contains a subterm that:$  $
\begin{enumerate}
	\item occurs in {\em operand} position in $e$,
	\item contains a parameter of order strictly less than its order.
\end{enumerate}
By ``operand position'' we mean ``in the second position of some
occurrence of the term application operator''.

A recursion scheme is said to be \textbf{safe} if it is not unsafe.
\end{definition}

The term `safety` comes from the fact that, when converted to lambda-terms, safe recursion schemes can be evaluated without ever having to generate a fresh variable during substitution: $\beta$-redexes can be reduced using (simultaneous) capture-permitting substitution \cite{BlumPhd,blumong:safelambdacalculus,safety-mirlong2004}.

The safe subset of the simply-typed lambda calculus introduced in \cite{blumong:safelambdacalculus,BlumPhd} relates to the notion of safe higher-order recursion schemes in the following sense:
\begin{proposition}[Safe rewrite rules and lambda terms {\cite[Proposition 3.11]{BlumPhd}}]
\label{prop:safelambdacorresp}
Take a recursion scheme $R = \langle \Sigma, \mathcal{N}, \mathcal{R}, S \rangle$.
For every non-terminal $N$, the associated rewriting rule is safe if and only if the simply-type term $\Lambda(N)$, with free variables in $\Sigma\union\mathcal{N}$, is derivable with the typing judgment of the safe lambda calculus of \cite{BlumPhd} and all types involved in the typing derivation are homogeneous.
\end{proposition}

It was shown that incremental-binding characterizes safe simply-typed lambda terms \cite[Proposition 5.11]{BlumPhd}. We show here the equivalent result for recursion schemes. Without loss of generality we will consider recursion schemes \defname{with no dead rule} such that for every production rule, there exists a derivation from the start non-terminal involving that rewrite rule.
\begin{proposition}[Binder characterization for HORS]
\label{prop:horsSafeBinderCharact}
A higher-order recursion scheme with no dead rule is \emph{safe} (in the sense of \cite{KNU02}) if and only if it is homogeneous and incrementally-bound.
\end{proposition}
\begin{proof}
The proof reduces to the setting of the lambda calculus where a similar result was shown for finite lambda terms \cite{BlumPhd}.
Take a safe recursion scheme $R = \langle \Sigma, \mathcal{N}, \mathcal{R}, S \rangle$.

$(\Longrightarrow)$ Observe that all the iterated unfoldings of $\Lambda(S)$ are safe homogeneous simply-typed terms. It's true of $\Lambda(S)$ itself by Proposition \ref{prop:safelambdacorresp}. And by the Substitution Lemma \cite[3.19]{BlumPhd}, it is also true of each subsequent unfolding. Consequently, by the characterization result of the safe-simply typed lambda calculus \cite[Proposition 5.11]{BlumPhd}(i), every repeated unfolding of $\Lambda(S)$ has an incrementally-bound computation tree. Hence $R$ is incrementally-bound.

$(\Longleftarrow)$ Assume that $R$ is not safe. Suppose that it's homogeneously-typed, then the syntactic constraint of safety is not met for some production rule $N x_1 \cdots x_n \rightarrow e \in \mathcal{R}$ and non-terminal $N \in \mathcal{N}$. Since the scheme has no dead-rule, after performing $|\mathcal{N}|$ unfoldings of $\Lambda(S)$ the non-terminal $N$ must have been substituted at least once by $\Lambda(N) = \lambda x_1 \cdots x_n. e$, which is unsafe by Proposition \ref{prop:safelambdacorresp}. Hence after a finite number of unfoldings of $S$ we obtain a term $t'$ containing the unsafe subterm $\Lambda(N)$.
Thus, by \cite[Proposition 5.11]{BlumPhd}(ii), the computation tree $t'$ is not incrementally-bound.
\end{proof}

\subsection{Equi-expressivity results}

Now that we have defined the notions of tree-generating higher-order recursion schemes and
tree-generating higher-order automata we can state two important results about their relative expressivity:

\begin{theorem}[Safe HORS - PDA equi-expressivity \cite{KNU02}]
A $\Sigma$-labelled tree is generated by a \emph{safe} order-$n$ recursion scheme if and only if it is accepted by an order-$n$ pushdown automaton.
\end{theorem}

\begin{theorem}[HORS - CPDA equi-expressivity \cite{hmos-lics08}]
	\label{thm:HMOS-equiexpr}
	A $\Sigma$-labelled tree is generated by an order-$n$ recursion scheme if and only if it is accepted by an order-$n$ collapsible pushdown automaton.
\end{theorem}
The proof of the HORS-CPDA equi-expressivity result is constructive in both direction: given a higher-order recursion scheme $R$, one can construct a collapsible pushdown automaton denoted $CPDA(G)$ that recognizes the same tree; and given a collapsible pushdown automaton $A$ one can construct a higher-order recursion scheme recognizing the same tree. In what follows, we will refer to them as the HMOS constructions.
\\

The following result summarizes the contribution of the present paper:
\begin{theorem}[Homogeneity is not a restiction]
	A $\Sigma$-labelled tree is generated by an \emph{incrementally-bound} order-$n$ recursion scheme if and only if it is accepted by an order-$n$ pushdown automaton.
\end{theorem}
Section \ref{sec:IncrementallyBoundRSToPDA} gives a constructive proof of the implication by showing that the collapsible pushdown automaton $CPDA(G)$ from the HMOS construction can be turned into an equivalent pushdown automaton (Theorem \ref{thm:correctness_simulation}).
Section \ref{sec:PDAtoIncrementallyBoundRS} proves the opposite direction.

\section{From recursion scheme to collapsible pushdown automaton}
In this section we recall some of the concepts from the HMOS constructions from Theorem \ref{thm:HMOS-equiexpr}. Familiarity with the original construction can be helpful, in particular for the concept of traversals. We refer the reader to \cite{hmos-lics08} for a more detailed introduction to those concepts.

We fix a higher-order recursion scheme $G = \langle \Sigma, \mathcal{N}, \mathcal{R}, S \rangle$ of order $n$. Let $\Nodes$ denotes the set of nodes of the computation tree of $G$; $\Nodes_@$ denotes the set of application nodes; $\Nodes_\lambda$ the set of lambda nodes; $\Nodes_\lambda^{\sf prime}$ the set of lambda nodes that are the first child of some $@$ nodes (also known as \emph{prime} lambda nodes); and $\Nodes_{\sf var}$ denotes the set of variables nodes.

\paragraph{Presentation}

\begin{definition}[CPDA of $G$]
We consider the order $n$ collapsible pushdown automaton, obtained from $G$ by the HMOS transformation defined in \cite{hmos-lics08} and denoted $CPDA(G)$ \cite[Definition 5.2]{hague-collaps-full}. Recall that:
\begin{itemize}
	\item The ranked-alphabet is $\Sigma$ (the alphabet of $G$);
	\item The stack-alphabet $\Gamma$ of the automaton is taken as the set of nodes of the computation \emph{graph} of $G$. Alternatively, this graph can be viewed as the computation tree of $\Lambda(S)$ by replacing back-pointers in the computation graph with pointers to non-terminals nodes in the computation tree. Thus $\Gamma = \Nodes$;
	\item The transition function $\delta$ of the automaton is shown in Figure \ref{fig:cpdaprime} using a more concise definition than the original one
	from \cite{hague-collaps-full};
	\item The initial configuration is defined as $c_0 = push_1\; \lambda\; \bot_n$
	where $\lambda$ refers to the root (dummy) lambda node of the computation graph of $G$.	
\end{itemize} 
\end{definition}
The automaton is well-defined in the sense that no collapse can occur at an element whose link is undefined: In particular $collapse$ never occurs at non-lambda nodes. It is equivalent to $G$ in the sense that they both generate the same tree \cite{hmos-lics08}.
The idea is that automaton $CPDA(G)$ proceeds by inductively computing a set of traversals over the computation tree of $G$. The traversals are shown to correctly evaluates the production rules of the recursion scheme $G$, therefore the constructed automaton correctly accepts the same tree as the recursion scheme. 

(At this point, readers not familiar with the concept of traversals over the computation graph of a recursion scheme, O-view and P-view may want to lookup their definitions in \cite{OngLics2006}.)

Observe that one can easily modify $CPDA(G)$ into an automaton that ``prints out'' the traversal that is being computed. This can be done by changing the behaviour of the $push_1$ operation to make it print out the input element before pushing it on the stack. The justification pointers can then be recovered inductively using the node labels: For a variable node, it is the only node-occurrence that binds it in the P-view at that point (which is computable by the induction hypothesis); prime lambda nodes always point to their immediate predecessor; and a non-prime lambda node $\lambda \overline{x}$ is always justified by the predecessor of the justifier of the variable node preceding it (written $\ip(\jp(t))$ where $t$ is the traversal ending with $\lambda \overline{x}$).

\begin{remark} 
Our presentation of $\delta$ differs slightly from the original one: In the case (A), when pushing the prime child of an application node $@$ on the stack, we assign it a link pointing to the preceding stack symbol in the top $1$-stack ({\it i.e.,} the $@$-node itself).
This modification avoids the case analysis on the value of $j$---the child-index of $u$'s binder---in the
cases ($V_0$) and ($V_1$), and the sequence of instructions $pop_1^{p+1}$ can just be replaced by
$pop_1^p ; collapse$. 
\end{remark}

The stack of the current configuration alone does not suffice to reconstruct the traversal that is being computed due to the use of the ``destructive'' operations $collapse$ in the CPDA. Nevertheless, two important pieces of information are recoverable from the configuration-stack: the O-view and the P-view of the traversal. 
\begin{proposition}[Recovering views from a configuration of $CPDA(G)$]
Let $c$ be a configuration of $CPDA(G)$. The \defname{long O-view},
\defname{O-view} and \defname{P-view} of the traversal currently being simulated by the configuration $c$, written respectively $\longoview{c}$, $\oview{c}$ and $\pview{c}$, can be retrieved using the following stack operations:
\begin{compactitem}
\item Long O-view:
\begin{align*}
&  \longoview{s} = \\
 & \quad \left\{
  \begin{array}{ll}
      \epsilon & \mbox{if $top_1 s$ is undefined;} \\
      \longoview{pop_1 s} \cdot top_1 s & \parbox[t]{7cm}{if $top_1\,s \in \Nodes_{\sf var}$, $top_1 s$ pointing to its immediate predecesor;} \\
      \longoview{pop_1 s} \cdot top_1 s & \mbox{if $top_1\,s \in \Nodes_@$, $@$ having no pointer;} \\
      \longoview{collapse\,s} \cdot top_1 s & \parbox[t]{7cm}{if $top_1\,s \in \Nodes_\lambda^{\sf prime}$, $top_1 s$ pointing to its immediate predecesor;} \\
      \longoview{collapse\,s} \cdot top_1 s & \parbox[t]{7cm}{if $top_1 s \in \Nodes_\lambda\setminus \Nodes_\lambda^{\sf prime}$, $top_1 s$ pointing to $\ip(\jp(s))$.}
    \end{array}
      \right.
\end{align*}
\item The O-view is defined similarly to the long-O-view except that the calculation stops when an @-node is reached:
\begin{align*}
  \oview{s}  &=   top_1 s & \mbox{if $top_1 s \in \Nodes_@$, $@$ having no pointer;}
\end{align*}
\item As shown in the HMOS construction, the P-view $\pview{s}$ is given by $top_2\,c$ \cite{hague-collaps-full}.
\end{compactitem}
\end{proposition}
\begin{proof}
	This follows from the inductive definition of traversals of $CPDA(G)$, $P$-views and $O$-views \cite{hague-collaps-full}.
\end{proof}

\begin{figure} 
\begin{center}
\makebox{
\begin{shadowbox}[10cm]
If $u$'s label is not a variables, the action is just a $push_1^v$, where $v$ is an appropriate child of the node $u$. Precisely:
\begin{itemize}
\item $(A)$ If the label is an @ then $\delta(u) = push_1^{E_0(u),{\bf 1}}$.
\item $(S)$ If the label is a $\Sigma$-symbol $f$ then $\delta(u) = push_1^{E_i(u)}$, where $1 \leq i \leq ar(f)$ is the direction requested by the Environment, or Opponent.

\item $(L)$ If the label is a lambda then $\delta(u) = push_1^{E_1(u)}$.
\end{itemize}
Suppose $u$ is a variable which is the $i$-parameter of
its binder and let $p$ be the span of $u$.
\begin{itemize}
\item $(V_1)$ If the variable has order $l\geq 1$, then
$$\delta(u) = push_{n-l+1} ; pop_1^p ; collapse;push_1^{E_i(top_1), n-l+1}$$
\item $(V_0)$ If the variable is of ground type then
$$\delta(u) = pop_1^p ; collapse;push_1^{E_i(top_1)}$$
\end{itemize}
\caption{The transition rules of $CPDA(G)$.}
\label{fig:cpdaprime}
\end{shadowbox}
}
\end{center}
\end{figure}

\paragraph{Reachable configurations}
We recall the notion of reachability with respect to the $\rightarrow$-step relation as defined in \cite{hague-collaps-full}: for two configurations $c$ and $c'$ we have $c\rightarrow c'$ just if $c' = \delta(top_1 c)(c)$ where $\delta$ is the transition function of the $CPDA(G)$.
In other words, a configuration is \defname{$\rightarrow$-reachable} if it can be attained starting from the initial configuration $c_0$ by performing
one or more applications of the steps (A), (S), (L), $(V_1)$, $(V_0)$ from the algorithm defining $CPDA(G)$.

The intermediate configurations visited by the internal transitions within a step are therefore not $\rightarrow$-reachable.
A configuration is said to be \defname{reachable} if it is $\rightarrow$-reachable or if it is
an intermediate configuration computed during a $\rightarrow$-step
({\it i.e.,} if it can be written $(op_1;\ldots;op_k)(c)$ where $c$ is $\rightarrow$-reachable and
$op_1, \ldots, op_k$ are the first $k$ instructions of some $\rightarrow$-step).

\paragraph{Link convention}
Observe that in $CPDA(G)$, when we push a lambda node $\lambda \overline{\xi}$ on the stack, the associated link has order $1$ if it is a prime lambda node (case (A)), and $n-\ord{\lambda \overline{\xi}}+1$ otherwise (case ($V_1$)). Hence, since no CPDA instruction can change the link order of an element pushed on the stack, at every stage during the execution of the CPDA the link order can be recovered from the (order of the) node itself.

From now on we will only work with stacks occurring as sub-stack of reachable configurations of $CPDA(G)$
therefore we will omit the order-component altogether when representing stack symbols: we write $\lambda\overline\xi^{k}$ to mean  $\lambda\overline\xi^{(1,k)}$ if $\lambda\overline\xi$ is a prime lambda node
and $\lambda\overline\xi^{(n-\ord{\lambda \overline{\xi}}+1,k)}$ otherwise.

\section{From incrementally-bound recursion schemes to pushdown automata}
\label{sec:IncrementallyBoundRSToPDA}
We now fix an {\bf incrementally-bound} higher-order recursion scheme $G$ of order $n$.
We first give a detailed analysis of $CPDA(G)$ and then show how to derive an equivalent (non-collapsible) order-$n$ pushdown automaton.

\subsection{Incremental order-decomposition}
\paragraph{Observation}
Let $s$ be a 1-stack. For any $l\in \nat$, $s$ can then be written
$$ s = u_{r+1} \cdot \lambda \overline{\eta}_r^{k_r} \cdot u_r \cdot
\ldots \cdot \lambda \overline{\eta}_1^{k_1} \cdot  u_1 $$
where
\begin{itemize}
\item  $\lambda \overline{\eta}_1^{k_1}$ is the
last $\lambda$-node in $s$ with order strictly greater than $l$;

\item for $1 < l \leq r$, $\lambda
\overline{\eta}_l^{k_l}$ is the last $\lambda$-node in $s_{\prefixof
\lambda \overline{\eta}_{l-1}^{k_{l-1}}}$ with order strictly
greater than $\ord{\lambda \overline{\eta}_{l-1}^{k_{l-1}}}$,

\item  $r$ is defined as the smallest number such that
$s_{\prefixof \lambda \overline{\eta}_{r}^{k_{r}}}$ does not contain
any lambda node of order strictly greater than $\lambda
\overline{\eta}_{r}^{k_{r}}$.
\end{itemize}

\noindent In other words:
\begin{itemize}
\item for $1 \leq k \leq r$, all the lambda nodes occurring in $u_l$ have order
strictly smaller than $\ord{\lambda \overline{\eta}_l}$;
\item for $1\leq l<l'\leq r$ we have $\ord{\lambda \overline{\eta}_l^{k_l}}
< \ord{\lambda \overline{\eta}_{l'}^{k_{l'}}}$;
\item $r=0$ if and only if all the lambda node in $s$ have order $\geq l$.
\end{itemize}

The subsequence $\lambda \overline{\eta}_r^{k_r} \ldots \lambda\overline{\eta}_1^{k_1}$ of $s$ consisting of the lambda nodes $\lambda
\overline{\eta}_l^{k_l}$ defined above is called the \defname{incremental order-decomposition of the $1$-stack} $s$ {\bf with respect to} $l\in\nat$.
This sequence is uniquely determined for any given $l\in \nat$.

\smallskip

We now generalize this notion to higher-order stacks.
\begin{definition}
The \defname{incremental order-decomposition of a higher-order stack} $s$ {\bf with respect to $l\in \nat$} (or
order-decomposition for short), written $\orddec_l(s)$, is defined as the order-decomposition of the top order-$1$ stack (defined using convention \ref{conv:top_preserve_links}). Equivalently it can be defined as follows:
\begin{align*}
  \orddec_l(\epsilon) &= \epsilon \\
    \mbox{for } s\neq\epsilon\quad \orddec_l(s) &=     \orddec_{\ord{\lambda\overline{\eta}}}(s_{<\lambda\overline{\eta}}) \cdot \lambda\overline{\eta}^k \\
& \mbox{\quad where } \lambda\overline{\eta}^k \mbox{ is the last node in $top_2\,s$ with order $>l$.}
\end{align*}
The \defname{incremental order-decomposition} of $s$, written
$\orddec(s)$, is defined as:
$$
\orddec(s) = \orddec_0(s) \ .
$$
\end{definition}

It follows immediately from the definition that:
\begin{equation}
l<l' \implies \orddec_{l'}(s) \prefixof \orddec_{l}(s)
\label{eqn_safe_prefix}\end{equation}
where $\prefixof$ denotes the sequence-prefix ordering.

\begin{lemma}
\label{lem:push1pop1_orderdecompo} Let $s$ be a (possibly higher-order) stack
such that $top_1\,s \in \Nodes_@ \union \Nodes_{var}$ and $l\geq 0$.
\begin{enumerate}[(i)]
\item Suppose that $\orddec_l(s) = \langle \lambda
\overline{\eta}_r^{k_r}, \ldots, \lambda \overline{\eta}_1^{k_1}
\rangle$ then for any lambda node $a \in \Gamma$ and link $(j,k)$ we have
 \begin{align*} &\orddec_l(push_1 a^{(j,k)}~s) = \\
&                 \qquad \left\{
                                       \begin{array}{ll}
\orddec_l(s), &\hbox{if $\ord{a}\leq l$;} \\
                                        \langle a^k \rangle, &  \hbox{if $\ord{a} \geq \ord{\lambda \overline{\eta}_r}$}; \\
                                         \langle \lambda \overline{\eta}_r^{k_r}, \ldots, \lambda
\overline{\eta}_i^{k_i}, a^k \rangle,
&                                        \mbox{otherwise,} \\
&                              i = \min \{ i \ \in \{1..r\}| \, \ord{a} <
\ord{\lambda \overline{\eta}_i} \}\ .
\\
                                       \end{array}
                                     \right.
\end{align*}
\item For any non-lambda node $a \in \Gamma$ and link $(j,k)$ we have
$$ \orddec_l(push_1\,a^{(j,k)}\,s) = \orddec_l(s) \ .$$

\item If $top_1 s$ is not a lambda node then
$$ \orddec_l(pop_1\,s) = \orddec_l(s) \ .$$
\end{enumerate}
\end{lemma}
\begin{proof}
Follows immediately from the definition of  $\orddec_l(s)$.
\end{proof}

\begin{lemma}[Incremental binders are in the order-decomposition]
\label{lem:binder_in_ordecompos} Let $c$ be a $\rightarrow$-reachable
configuration of $CPDA(G)$ such that $top_1\,c$ is a variable $x$. Then
\begin{enumerate}[(i)]
\item $\orddec(c)$ contains at least a node with order strictly greater than $\ord(x)$;
\item the last lambda node in $\orddec(c)$ satisfying the first condition is precisely $x$'s binder.
\end{enumerate}
In other words, $x$'s binder is the last lambda node in $\orddec_{\ord x}(c)$.
\end{lemma}
\begin{proof}
\begin{enumerate}[(i)]
\item The top $1$-stack of a $\rightarrow$-reachable configuration contains the P-view of some traversal whose last visited node is the $top_1$ symbol \cite[Corollary 8]{hague-collaps-full}; and
    the P-view of a traversal is exactly the path (in the unfolding of) the
    computation graph from the last visited node to
    the root \cite[Proposition 6]{OngLics2006}. Hence the binder of $x$, whose order
    is strictly greater than $\ord{x}$, occurs in the top $1$-stack.
    Consequently $\orddec(c)$ must contain at least one node of order strictly greater than $l$.

\item Since the recursion scheme is incrementally-bound, $x$'s binder is precisely the first $\lambda$-node in the
 path to the root in the computation tree with order strictly
 greater than $x$. \qedhere
\end{enumerate}
\end{proof}

\subsection{Stack safety}

\begin{definition}[Safe stack]
\label{dfn:safestack} Let $s$ be an order-$j$ non-empty stack for $j\geq1$.

The stack $s$ is \defname{$l$-safe} iff
    \begin{enumerate}[1.]
    \item $\orddec_l(s) = \langle \lambda \overline\eta_r^1, \ldots ,
    \lambda \overline{\eta}_1^1 \rangle$ for some $r\geq 0$ {\it i.e.}, the height component of the links in $\orddec_l(s)$ are all equal to $1$;

    \item for all $1 \leq q \leq r$ such that $n-\ord{\lambda \overline\eta_q}+1 \leq \ord{s}$:
    \begin{itemize}
    \item for $q = 1$ we have $collapse~s_{\prefixof \lambda\overline\eta_1}$ is $l$-safe;
    \item for $q>1$ we have $collapse~s_{\prefixof \lambda\overline\eta_q}$ is $\ord(\lambda \overline\eta_{q-1})$-safe.
    \end{itemize}
    \end{enumerate}

We say that $s$ is \defname{safe} if it is $0$-safe.
\end{definition}

Since $s$ is a stack, and not necessarily a configuration, it may have dangling pointers.
The condition $n-\ord{\lambda \overline\eta_q}+1 \leq \ord{s}$ in the definition ensures that
$\lambda\overline\eta_j$'s link is not dangling so that we can indeed perform a collapse at $\lambda\overline\eta_j$.

%
%

\begin{lemma}
\label{lem:stacksafety_immediate results}
Let $s$ be a stack such that $\orddec_l(s) = \langle \lambda \overline\eta_r^{k_r}, \ldots, \lambda \overline{\eta}_1^{k_1} \rangle$ and $l\geq 0$.
If $s$ is $l$-safe and $l \leq k \leq n$ then $s$ is $k$-safe.
\end{lemma}
\begin{proof}
Immediate consequence of the definition.
\end{proof}

\begin{lemma}[Collapse simulation]
\label{lem:safecollapsesimulation}
Let $s$ be a sub-stack of a reachable configuration of $CPDA(G)$ and $l\geq 0$.
If $\ord s \geq 2$ and $top_2\,s$ is $l$-safe or if $\ord s = 1$ and $s$ is $l$-safe then for any lambda node $\lambda\overline\eta$ in $\orddec_l(s)$ we have:
$$collapse\,s_{\prefixof\lambda\overline\eta} =
\left\{
\begin{array}{ll}
pop_1\,s_{\prefixof\lambda\overline\eta} & \mbox{if $\lambda\overline\eta$ is prime,}\\
pop_{n-\ord \lambda\overline\eta +1}\,s_{\prefixof\lambda\overline\eta} & \mbox{otherwise.}
\end{array}
\right.
$$
\end{lemma}
\begin{proof}
The $collapse$ operation is defined as  $collapse\,s = pop_j^k\,s$ where $(j,k)\in\nat\times\nat$ is the link attached to $top_1\,s$.
Since $s$ is a sub-stack of a reachable configuration we have $j=1$ if $\lambda\overline\eta$ is prime and $j = n-\ord \lambda\overline\eta +1$ otherwise. Furthermore, since $top_2\,s$ is safe and $\lambda\overline\eta$ belongs to the order-decomposition, the height component necessarily equals $1$.
\end{proof}

\subsubsection{Operations preserving stack safety}
\begin{lemma}
\label{lem:push1pop1_preserves_safety} Let $s$ be a higher-order
stack. Suppose that $s$ is $l$-safe, $l\geq0$. Then:
    \begin{enumerate}[(i)]
    \item $top_{\ord{s}}\,s$ is $l$-safe;
    \item if $top_1\,s$ is not a lambda node then $pop_1\,s$ is $l$-safe;
    \item for every non-lambda node $a$, $1\leq j\leq n$, $k\geq 1$, $push_1\,a^{(j,k)}\,s$ is $l$-safe;
    \item for every lambda node $a$, $push_1\,a^{(1,1)}\,s$ is $l$-safe if $\ord a < l$, and safe if $\ord a \geq l$.
\end{enumerate}
\end{lemma}
\begin{proof}
This is a direct consequence of Lemma \ref{lem:push1pop1_orderdecompo}.
For (vi), the cases $\ord a > l$ and $\ord a < l$ follow immediately from Lemma \ref{lem:push1pop1_orderdecompo}(i);
for $\ord a = l$ it follows from the fact that $\orddec_0 (push_1\,a^{(1,1)}\,s) = \orddec_l (s) \cdot a^1$.
\end{proof}

\begin{lemma}
\label{lem:incrk_qsafe}
Let $0\leq l < n$, $q\geq 0$ and $s$ be a stack of level $1\leq \ord{s} <n$.
If $s$ is $q$-safe then $s^{\langle n-l+1 \rangle}$ is $\max(l,q)$-safe.
\end{lemma}
\begin{proof}
Let $s$ be a safe stack with $1\leq \ord s <n$. We prove the result by induction on the size of $s$.
The base case is the trivial: $s$ is the empty stack.
Step case: Since $s$ is $q$-safe we have
$\orddec_q(s) = \langle \lambda \overline{\eta}_r^1
, \ldots, \lambda \overline{\eta}_1^1   \rangle$.
By (\ref{eqn_safe_prefix}), $\orddec_{\max(l,q)}(s)$ is a prefix of $\orddec_q(s)$.
Let $b$ be the index in $\orddec_q(s)$ of the last node of $\orddec_{\max(l,q)}(s)$:
thus $\lambda\overline\eta_b$ is the last lambda node in $top_2\, s$ with order $>\max(l,q)$.

The stack-operation $(\cdot)^{\langle n-l+1 \rangle}$ updates the pointers as follows:
the height component of each link is incremented if the order of the stack symbol is $l$ and is kept unchanged otherwise.
Hence we have:
\begin{equation*}
\orddec_q(s^{\langle n-l+1 \rangle}) = \langle
\lambda \overline{\eta}_r^1, \ldots,  \lambda \overline{\eta}_{b}^1, \lambda \overline{\eta}_{b-1}^{k}, \lambda \overline{\eta}_{b-2}^1 \ldots,
 \lambda \overline{\eta}_1^{1} \rangle
\end{equation*}
for some $1 \leq k\leq 2$. And therefore:
\begin{equation}
\orddec_{\max(l,q)}(s^{\langle n-l+1 \rangle}) = \langle \lambda \overline{\eta}_r^1
, \ldots, \lambda \overline{\eta}_b^1 \rangle\ . \label{eqn:trivialityofcollapse}
\end{equation}

Now consider an index $j$ such that $b\leq j \leq r$ and $n-\ord{\lambda\overline\eta_j} +1 \leq \ord s$.
Since the height component of $\lambda\overline\eta_j$'s link is not affected by the operation
$(\cdot)^{\langle n-l+1 \rangle}$, this operation commutes with $collapse$ and we have:
\begin{equation}
 collapse\, s^{\langle n-l+1 \rangle}_{\prefixof \lambda \overline{\eta}_{j}}
= (collapse\, s_{\prefixof \lambda \overline{\eta}_{j}})^{\langle n-l+1 \rangle} \ . \label{eqn:commutecollapse}
\end{equation}

By assumption $s$ is $q$-safe therefore $collapse\,s_{\prefixof \lambda \overline{\eta}_{j}}$ is $q$-safe
if $j=b=1$ and $\ord(\lambda\overline\eta_{j-1})$-safe if $j>b\geq1$.

Since $collapse\,s_{\prefixof \lambda\overline\eta_j}$ is strictly smaller than $s$, by the induction hypothesis we have that
$(collapse\,s_{\prefixof \lambda\overline\eta_j})^{\langle n-l+1\rangle}$ is $\max(q,l)$-safe
if $j=b=1$, and $\max(\ord(\lambda\overline\eta_{j-1}),l)$-safe if $j\geq b>1$.

For $j>b$, $\lambda\overline\eta_{j-1}$ occurs in $\orddec_l\, s$ therefore $\ord(\lambda\overline\eta_{j-1}) > l$,
similarly for $j=b$ we have $\ord(\lambda\overline\eta_{j-1}) \leq l$.
Hence $(collapse\,s_{\prefixof \lambda\overline\eta_j})^{\langle n-l+1\rangle}$ is
\begin{enumerate}[(i)]
\item $\max(q,l)$-safe for $j=b=1$,
\item $l$-safe for $j=b>1$, and therefore $\max(q,l)$-safe by Lemma \ref{lem:stacksafety_immediate results},
\item $\lambda\overline\eta_{j-1}$-safe for $j>b$.
\end{enumerate}
This shows that $(collapse\,s_{\prefixof \lambda\overline\eta_j})^{\langle n-l+1\rangle}$ is $\max(l,q)$-safe, and therefore by (\ref{eqn:commutecollapse})
so is $collapse\, s^{\langle n-l+1 \rangle}_{\prefixof \lambda \overline{\eta}_{j}}$.
\end{proof}

\begin{lemma}
\label{lem:cons_qsafety} Let $s$ be a higher-order stack of level $\geq 2$ and $l\geq0$.
 If
\begin{enumerate}[1.]
\item $pop_{\ord{s}}\,s$ is safe,
\item and $top_{\ord{s}}\,s$ is $l$-safe,
\end{enumerate}
then $s$ is $l$-safe.
\end{lemma}
\begin{proof}
Let $s = [s_1 \ldots s_r~s_{r+1}]$ for some $l\geq0$.
We proceed by induction on  $top_{\ord{s}}\,s=s_{r+1}$.
The base case $s_{r+1} = \bot_{\ord{s}-1}$ is trivial.
Suppose that $s_{r+1}$ is not the empty stack.

\begin{enumerate}[(i)]
\item Clearly $\orddec_l\, s = \orddec_l\, s_{r+1}$, hence since $s_{r+1}$ is $l$-safe the lambda nodes in $\orddec_l\, s$ have all a link of height $1$.

\item Let $\lambda \overline{\eta}$ be a lambda node in $\orddec_l(s) = \orddec_l(s_{r+1})$ such that $n-\ord{\lambda \overline{\eta}}+1 \leq \ord{s}$.
Since its link is of height 1 we have
$collapse~s_{\prefixof \lambda \overline{\eta}} = pop_{n-\ord{\lambda \overline{\eta}}+1}~s_{\prefixof \lambda \overline{\eta}}$.

If $n-\ord{\lambda \overline{\eta}}+1 = \ord{s}$ then
$pop_{n-\ord{\lambda \overline{\eta}}+1}~s_{\prefixof \lambda \overline{\eta}} = pop_{\ord{s}}~s_{\prefixof\lambda\overline\eta} = pop_{\ord{s}}~s$ which is $l$-safe by the first assumption
and Lemma \ref{lem:stacksafety_immediate results}.

Otherwise $n-\ord{\lambda \overline{\eta}}+1 < \ord{s}$ and we have:
\begin{align*}
 & collapse\,s_{\prefixof\lambda\overline\eta} \\
&  \qquad    = pop_{n-\ord{\lambda \overline{\eta}}+1}~s_{\prefixof\lambda\overline\eta}
              & \mbox{since $\lambda\overline\eta$'s link has height 1} \\
 & \qquad  = [ s_1 \ldots s_l\ (pop_{n-\ord{\lambda\overline\eta}+1}~ {s_{r+1}}_{\prefixof\lambda\overline\eta}) ]
           & \mbox{$n-\ord{\lambda \overline{\eta}}+1<\ord{s}$}  \\
&  \qquad  = [s_1 \ldots s_p\ (collapse\, {s_{r+1}}_{\prefixof\lambda\overline\eta}) ]
          & \mbox{since $\lambda\overline\eta$'s link has height 1.}
\end{align*}

By the second assumption, $s_{r+1}$ is $l$-safe therefore $collapse~{s_{r+1}}_{\prefixof\lambda\overline\eta}$
is $l$-safe if $\lambda\overline\eta$ is the last node in the $l$-order decomposition,
and $k$-safe where $k$ is the order of the following node in $\orddec_l(s_{r+1})$ otherwise.

Since $|collapse~{s_{r+1}}_{\prefixof\lambda\overline\eta}| < |s_{r+1}|$ we can use the induction hypothesis
to show that the same hold for $[s_1 \ldots s_p\ (collapse\, {s_{r+1}}_{\prefixof\lambda\overline\eta}) ]$. Therefore it is $l$-safe and so is
$collapse\,s_{\prefixof\lambda\overline\eta}$ by the previous equality.
\qedhere
\end{enumerate}
\end{proof}

\begin{lemma}
\label{lem:pushj_safe_implies_l-safe} Let $n>l\geq 1$ and $s$ be a safe higher-order stack such that $2 \leq n-l+1 \leq \ord{s} \leq n$. Then $push_{n-l+1}\ s$ is $l$-safe.
\end{lemma}
\begin{proof}
Let $s=[s_1 \ldots s_{c+1}]$ be a safe higher-order stack such that $2 \leq n-l+1 \leq \ord{s} \leq n$. Then by Lemma \ref{lem:push1pop1_preserves_safety}(i), $s_{c+1}$ is safe.

We show that $push_{n-l+1}~s$ is $l$-safe by finite induction on the order of $s$.
    \begin{compactitem}
      \item Base case: $\ord{s} = n-l+1 $. We have
    $push_{n-l+1}~s = [ s_1 \ldots s_{c+1} s_{c+1}^{\langle n-l+1
    \rangle}]$.

    Since $s_{c+1}$ is safe,  by Lemma \ref{lem:incrk_qsafe} $s_{c+1}^{\langle n-l+1\rangle}$ is $l$-safe, and by Lemma
    \ref{lem:cons_qsafety},  $[ s_1 \ldots s_{c+1} s_{c+1}^{\langle n-l+1
    \rangle}]$ is $l$-safe.

      \item Suppose $\ord{s} > n-l+1$. Then
    $push_{n-l+1}~s = [ s_1 \ldots s_{c+1} push_{n-l+1}\,s_{c+1}]$.

    Since $s_{c+1}$ is safe, by the
    induction hypothesis $push_{n-l+1}~s_{c+1}$ is $l$-safe, and by Lemma \ref{lem:cons_qsafety} so is $[ s_1 \ldots s_{c+1} push_{n-l+1}\,s_{c+1}]$.
\qedhere
    \end{compactitem}
\end{proof}

\subsection{Simulation and proof of correctness}
\begin{proposition}
\label{prop:reachconf_safe}
Let $G$ be an incrementally-bound recursion scheme.
The $\rightarrow$-reachable configurations of $CPDA(G)$ are safe.
\end{proposition}
\begin{proof}
If $n =\ord{c} =1$ then the result holds trivially since $CPDA(G)$ does not contain
any transition of the form $push_j$ for $j>1$ and therefore the links in a reachable configuration all have a height component equal to $1$.

Take $n\geq 2$. We proceed by induction on the $\rightarrow$-step relation. The initial configuration is clearly safe.
Suppose that $c$ is a safe $\rightarrow$-reachable configuration and that
$c \rightarrow c'$. We do a case analysis on $top_1\,c$:
\begin{itemize}
\item (A): We have $c'=push_1^{E_0(u),1}\,c =push_1\,E_0(u)^{(1,1)}\,c$ where $E_0(u)$ denotes a lambda node. It is safe by Lemma \ref{lem:push1pop1_preserves_safety}(iv).

\item (S):
We have $c' = push_1^{a}\,c = push_1\,a^{(j,k)}\,c$ for some dummy lambda node $a$
and undetermined link $(j,k)$. It is safe by Lemma \ref{lem:push1pop1_preserves_safety}(iv) since $\ord a = 0$.

\item (L): We have $c' = push_1^{E_1(u)} = push_1 E_1(u)^{(j,k)}$ where $E_1(u)$ is not a lambda node
 and $j,k\geq1$ are undetermined. It is safe by Lemma \ref{lem:push1pop1_preserves_safety}(iii).

\item ($V_1$) \& ($V_0$): Suppose $u$ is labelled by a variable $x$ of order $l$.
Since $c$ is safe we have $\orddec(c) = \langle \lambda\overline\eta_r^1 , \ldots, \lambda\overline\eta_1^1
\rangle$, $r\geq 0$. Since the recursion-scheme $G$ is safe, by Lemma \ref{lem:binder_in_ordecompos}, $x$'s binder is precisely the last node of $\orddec_l(c)$. Let $b$ be its index in $\orddec(c)$, and $i\geq 1$ be the index of $x$
in $\overline\eta_b$.

\begin{compactitem}
\item  ($V_1$): $l\geq 1$.
We have $c'= push_1\,E_i(top_1)^{(n-l+1,1)}\, t$ where
$t$ is given by $(push_{n-l+1};pop_1^p;collapse)(c) = collapse ((push_{n-l+1}\,c)_{\prefixof \lambda
\overline{\eta}_{b}} )$.
By Lemma \ref{lem:pushj_safe_implies_l-safe} $push_{n-l+1}\, c$ is $l$-safe therefore, since $\lambda\overline\eta_b$ is the last node in $\orddec_l(c)$, by definition of $l$-safety we have that $t$ is $l$-safe. Finally the lambda node $E_i(top_1)$ pushed by the last operation has precisely order $l =\ord(x)$ therefore
\begin{align*}
\orddec_0 (c') &= \langle \lambda\overline\eta_r^1 , \ldots, \lambda\overline\eta_b^1, E_i(top_1(t))^1\rangle \ .
\end{align*}
Thus all the lambda nodes in $\orddec_0 (c')$ have a link of height $1$.

We now need to show that safety is preserved when collapsing at nodes of $\orddec_0 (c')$. Let $b\leq j \leq r$, we have $c'_{\prefixof \lambda\overline\eta_j} =
t_{\prefixof \lambda\overline\eta_j}$. For $j>b$, the l-safety of $t$ implies that $collaspse\, c'_{\prefixof \lambda\overline\eta_j}$ is $\ord{\lambda\overline\eta_{j-1}}$-safe as required. For $j=b$ it gives that $collaspse\, c'_{\prefixof \lambda\overline\eta_b}$ is $l$-safe as required since $l = \ord{E_i(top_1(t))}$.

Now it remains to show that $collapse(c'_{\prefixof{E_i(top_1(t))}}) = collapse\, c'$ is safe.
Since we have $i\geq 1$ the node $top_1(c') = E_i(top_1(t))$ is not a prime lambda node, thus by Lemma \ref{lem:safecollapsesimulation} we can simulate the $collapse$ by a $pop$ of order $n-\ord{E_i(top_1(t))}+1$:
\begin{align*}
collapse\, c' &=  pop_{n-\ord{E_i(top_1(t))}+1}\,c'  \\
               &= pop_{n-l+1}\, c' \\
               &= (push_{n-l+1};pop_1^p;collapse;push_1^{E_i(top_1),n-l+1};pop_{n-l+1})\,c
\end{align*}
The operation $pop_1$ and $push_1$ only affects the top $1$-stack.
Furthermore, since $x$'s binder has order $>l$, its link has order $<n-l+1$ therefore
the collapse operation following the $pop_1^p$ only affects the
top $(n-l)$-stack. Consequently, the operation $pop_{n-l+1}$ effectively restores the configuration
to its value prior to performing the $push_{n-l+1}$ operation:
$$collapse\, c' = c \ .$$
Hence $c'$ is safe.

\item ($V_0$): $l=0$ (which implies that  $b=1$). The configuration $c'$ is given by
$push_1 E_i(top_1)\,collapse(c_{\prefixof \lambda \overline\eta_b})$.
Since $c$ is safe by definition so is $collapse(c_{\prefixof \lambda \overline\eta_b})$.
Since the pushed lambda node $E_i(top_1)$  has order $l=0$, by
Lemma \ref{lem:push1pop1_preserves_safety}(iv) $c'$ is safe.
\qedhere
\end{compactitem}
\end{itemize}
\end{proof}

\begin{definition}[Simulating PDA]
\label{simulating_pda}
Let $G$ be an incrementally-bound recursion scheme.
We define $PDA(G)$ as the higher-order PDA obtained from
$CPDA(G)$ by replacing every $collapse$ operation by $pop_1$ if $top_1\, s$ is prime, and by $pop_{n-\ord{top_1(s)}+1}$ otherwise.
\end{definition}

\begin{theorem}[Correctness of the simulation]
\label{thm:correctness_simulation}
$PDA(G)$ and $CPDA(G)$ are equivalent.
\end{theorem}
\begin{proof}
In $CPDA(G)$, the collapse operation occurs only in the steps ($V_1$) and ($V_0$). For ($V_1$) it is of the form:
$$collapse(pop_1^p(push_{n-l+1}~c))$$
for some $\rightarrow$-reachable configuration $c$, where $top_1\,c$ is a variable $x$ of order $l$ and span $p$.
By the previous proposition, $c$ is safe and by Lemma \ref{lem:pushj_safe_implies_l-safe} $push_{n-l+1}\, c$ is $l$-safe. Since $x$ has span $p$, after the operation $pop_1^p$ the top stack symbol is precisely $x$'s binder, which by Lemma \ref{lem:binder_in_ordecompos}, belongs to $\orddec_l(c)$, therefore
by Lemma \ref{lem:safecollapsesimulation} the collapse can be soundly simulated by a $pop$ of order
$1$ if $x$'s binder is a prime node, and a pop of order $n-\ord(top_1(pop_1^p(push_{n-l+1}\,c))) +1$ otherwise.

The case ($V_0$) is proved similarly.
\end{proof}

\begin{remark}
The result also holds for the slightly different definition of the automaton $CPDA(G)$ from the original HMOS transformation \cite{hmos-lics08}:
Clearly the two CPDAs have the same set of $\rightarrow$-reachable configurations
so Proposition \ref{prop:reachconf_safe} clearly still holds.
The simulating PDA from Def.\ \ref{simulating_pda}, however, is obtained 
by replacing every $collapse$ operation in the transition of the CPDA by $pop_{n-\ord{top_1(s)}+1}$. Also the equality in Lemma \ref{lem:safecollapsesimulation} becomes:
$$collapse\,s_{\prefixof\lambda\overline\eta} = pop_{n-\ord \lambda\overline\eta +1}\,s_{\prefixof\lambda\overline\eta} \ .$$
The correctness of the simulation follows similarly.
\end{remark}

\section{From PDA to incrementally-bound RS}
\label{sec:PDAtoIncrementallyBoundRS}
\begin{proposition}
	Let $\mathcal{A}$ be an order-$n$ PDA, and $\Sigma$ a ranked alphabet. There exists an order-$n$ incrementally-bound recursion scheme $R$
	such that for every $\Sigma$-labelled tree $t$, $t$ is accepted by $\mathcal{A}$ if and only if it is generated by $R$.
\end{proposition}
\begin{proof}
Let $\mathcal{A}$ be an order-$n$ PDA. Following the HMOS transformation from order-$n$ CPDA to order-$n$ recursion scheme \cite{hmos-lics08} we show how to produce an equivalent homogeneously-typed recursion scheme that is incrementally-bound.

Let's consider the PDA $\mathcal{A}$ as a CPDA. The HMOS transformation yields an equivalent recursion scheme
$R = \langle \Sigma, \mathcal{N}, \mathcal{R}, S \rangle$.
(Note: we refer the reader to \cite{hmos-lics08} for the definition of the production rules
$\mathcal{R}$.)

Since a PDA, viewed as a CPDA, never makes used of the $collapse$ operation, we can get rid of the parameters $\overline\Phi$ from the
HMOS construction, used to implement the collapse stack operation via production rules.
With this simplification, the type of the non-terminal $\mathcal{F}_p^{a,e}$  for each stack symbol $a$, $1\leq e \leq n$, and state $1\leq p \leq m$  becomes:
 $$ \mathcal{F}_p^{a,e} : (n-1)^m \typear \ldots \typear 0^m \typear 0$$

which is homogeneously-typed.

The production rules get simplified to the general form:
 $$ \mathcal{F}_p^{a,e} \  \overline{\Psi_{n-1}} \ldots \overline{\Psi_0} \stackrel{(q,\theta)}\rightarrow \Xi_{(q,\theta)}$$
for every state $q$ and stack operation $\theta$ of the PDA's transition function,
where the right-hand side $\Xi_{(q,\theta)}$ is given by the following table:
$$\begin{array}{ll}
\mbox{\bf Operation $\theta$} & \mbox{\bf Applicative term } \Xi_{(q,\theta)} \\[10pt]
push_1^{b,k} & \mathcal{F}_q^{b,k} \langle \mathcal{F}_i^{a,e} \overline{\Psi_{n-1}} | i \rangle \overline{\Psi_{n-2}} \ldots
 \overline{\Psi_0} \\[10pt]
push_j & \mathcal{F}_q^{a,e} \overline{\Psi_{n-1}} \ldots \overline{\Psi_{n-(j-1)}}\langle \mathcal{F}_i^{a,e} \overline{\Psi_{n-1}} \ldots \overline{\Psi_{n-j}} | i \rangle \overline{\Psi_{n-(j+1)}} \ldots
 \overline{\Psi_0} \\[10pt]
pop_k & \Psi_{n-k,q} \overline{\Psi_{n-k-1}} \ldots \overline{\Psi_{0}}
 \end{array}
 $$

It is easy to verify that every application term in the above table obeys the syntactic restriction of the safety constraint from Def.\ \ref{def:safers}.
Hence since the recursion scheme is homogeneously-typed, by Proposition \ref{prop:horsSafeBinderCharact}, it is also incrementally-bound.
\end{proof}

\bibliographystyle{abbrv}
\bibliography{dphil-all}

\end{document}